%% file: ms.tex
\documentclass[sigconf]{acmart}

\usepackage{amsmath}
\usepackage{algorithm}
\usepackage[noend]{algpseudocode}
\usepackage{tikz}
\usetikzlibrary{er,positioning}

\makeatletter
\def\BState{\State\hskip-\ALG@thistlm}
\makeatother

\begin{document}
\title{Equivalence Test in Multi-dimensional Space}
\subtitle{with Applications in A/B Testing}

\author{Jing Miao}
\affiliation{%
  \institution{Stanford University}
  \city{Stanford} 
  \state{California} 
  \postcode{94305}
}
\email{jingm@stanford.edu}

\author{Hongyuan Yuan}
\affiliation{%
  \institution{Adobe}
  \city{San Jose} 
  \state{California} 
  \postcode{43017-6221}
}
\email{hoyuan@adobe.com}

\author{Zhenyu Yan}
\affiliation{%
  \institution{Adobe}
  \city{San Jose} 
  \state{California}}
\email{wyan@adobe.com}


\begin{abstract}
In this paper, we provide a statistical testing framework to check whether a random sample splitting in a multi-dimensional space is carried out in a valid way, which could be directly applied to A/B testing and multivariate testing to ensure the online traffic split is truly random with respect to the covariates. We believe this is an important step of quality control that is missing in many real world online experiments. Here, we propose a randomized chi-square test method, compared with propensity score and distance components (DISCO) test methods, to test the hypothesis that the post-split categorical data sets have the same multi-dimensional distribution. The methods can be easily generalized to continuous data. We also propose a resampling procedure to adjust for multiplicity which in practice often has higher power than some existing method such as Holm's procedure. We try the three methods on both simulated and real data sets from Adobe Experience Cloud and show that each method has its own advantage while all of them establish promising power. To our knowledge, we are among the first ones to formulate the validity of A/B testing into a post-experiments statistical testing problem. Our methodology is non-parametric and requires minimum assumption on the data, so it can also have a wide range of application in other areas such as clinical trials, medicine, and recommendation system where random data splitting is needed.
\end{abstract}

%
%


\keywords{A/B testing, online traffic split, distribution test, multi/high-dimensional data, randomized Chi-square test, propensity score, DISCO}

\maketitle

\input{samplebody-conf}


\end{document}

%% file: samplebody-conf.tex
\section{Introduction}
A/B testing (also known as split testing or bucket testing) is a method of comparing two versions of a webpage or app against each other to determine which one performs better. It is essentially a special case of multiple variants of a page which are shown to users at random, and statistical analysis is used to determine which variation performs better for a given key performance indicator (KPI). Sometimes, the results from A/B testing are used as a baseline to be compared with other digital marketing services, e.g. personalization. Thus, it is important that A/B testing or multi-variate testing results are truly valid so that we make the correct inference. One important aspect is to get rid of any confounding. In other words, we want to make sure that the randomly assigned groups have the same profile distributions. In this way any difference in the KPI actually results from the
testing objects, instead of from some characteristics of the users such as age or gender. Despite the large volume of literature on A/B testing design, to our knowledge, there has not been any papers which aims to validate the balance of population distribution of different groups. We want to point out that our distribution testing methods are non-parametric by design and can have much more general applications beyond A/B testing, such as in randomized medical clinical trials, in design of experiments in manufacturing fault detection, and in bandits problem in personalized medicine, basically anywhere when uniformly random data splitting is required in a multi-dimensional space.

\par The rest of this paper has the following structure: Section \ref{relatedwork} discusses past literature on theory and application of randomized controlled trials (RCT) and observational studies. Section \ref{methodology} introduces some statistical tools, i.e., distance-covariance analysis, propensity score method, randomized chi-square test, and a resampling technique, to test equivalence of multi-dimensional categorical distribution. Section \ref{simulations} compares and analyzes the above methods on simulated data. Section \ref{realdata} applies our methods on real traffic data of online marketing from Adobe Experience Cloud. Section \ref{discussion} discusses further generalization and implementations of our methodology.  

\section{Past Work}
\label{relatedwork}
Randomized controlled trials are the simplest method to yield unbiased estimates of the treatment effects under a minimal set of assumptions. Researchers have been developing new methodology that boosts the efficiency of randomized trials when sample size is small.  
Sample size may not be a problem in some e-commerce settings as the number of online visitors can easily build up to hundreds of thousands. However, in other settings such as clinical trials, the number of samples (patients) is often small \cite{bhat2017}. \cite{bhat2017} formulates online/offline A-B testing as a (computationally challenging) dynamic optimization problem and develops approximation and exact algorithms. In particular, the paper assumes that response is linear in the treatment and covariates: $y_k = x_k \theta + Z_k \kappa + \epsilon_k$, where $x_k=0$ or $1$ according to whether subject $k$ is assigned to treatment or control group, and $Z_k$ denotes covariates of subject $k$. The objective is to maximize overall precision, which is the inverse of the standard error of $\hat{\theta}$. In online setting, $x_k$ is $\mathcal{F}_{k-1}$ measurable. It can also incorporates other goals such as controlled selection bias and endogenous stopping. Despite clean theoretical guarantees and relatively easy implementation of this algorithm, the linear model is not applicable in many real cases, for example, if $y_k$ is 0 or 1 denoting whether a user convert or not. In this case, a generalized linear model makes more sense but we are unclear about any theoretical results in this scenario. Also the theory requires $Z_k$ to have elliptical distribution, which is not the case in many real applications, though one real experiment with categorical data in the paper still performs well. More importantly, the number of covariates may not be clear at the beginning of A/B testing as online user profile may not be available at the time of experiments. Since the current paper is concerned with experiments in e-commerce, the sample size does not pause a big issue.

\par There are numerous works on the implementation, validity, and efficiency of A/B testing. \cite{kohavi2009} gives a detailed summary of the things data scientists need to pay attention to when carrying out A/B testing, and in general multivariate testing, such as determining the right sample size before testing and stopping the experiments when buggy or unintentionally poor results come up. When there are more than one variant to test, such as the choice of the picture or the color of the text on a website, online experimenters can either do multiple pairwise comparison of each individual variant, and apply methods of multiple testing adjustments, such as Bonferroni adjustment, if the number of variants is small, or do a multivariate test altogether. This paper mentions some possibilities that A/B testing split might not be truly random. One of the most common in practice is the effect of robots. Robots can introduce significant skew into estimates, enough to render the split invalid, and cause many metrics to be significant when they should not have been. For some websites robots are thought to provide up to half the pageviews on the site \cite{kohavi2004}. However, it is difficult to clearly delineate between the human users and robots \cite{tan2002}. Despite mentioning possible confounding factors in online experiments, this paper does not go into details on how to get rid of them, and how to test whether the splitting of traffic is truly random.

\par When randomized trials are not possible, techniques in observational studies are needed to get valid inference. \cite{gordon2016} focuses on the evidence of advertising measurements from big field experiments in Facebook. It uses the advertising effects measured from RCTs as benchmark, and compares the causal effects measured from observational studies to the benchmark. Two methods proposed by the paper to get rid of confounding in non randomized experiments are matching and regression adjustments. Exact matching and propensity matching are two common ways in practice for the former. Inverse-probability-weighed regression adjustment (IPWRA) \cite{wooldridge2007} incorporates propensity information into regression adjustments to provide more robust results. This paper also mentions checking the equivalence of distribution in A/B split so as to create valid benchmark. Yet it simply checks each covariate independently without adjusting for multiplicity. This can incur false positives when the number of covariates is large and cannot detect difference in covariance structure. A/B testing is a design experiment so we do not need to use the methodology of observational studies in our paper. 

\par There are also papers on the continuous monitoring of A/B testing and how to effectively run a large number of A/B testing. If we plot online experiments results continuously and declare significance the first time p-value is smaller than a pre-defined threshold, we will find that the actual type I error can be much larger than the threshold. \cite{johari2015} introduces methods to continuously monitoring A/B while still controlling type I error or the false discovery rate. \cite{tang2010} describes overlapping experiments infrastructure that helps run experiments faster and produce better decisions. \cite{zhang2017} describes an adaptive sample size modification for cold start of A/B testing. Network effect can also come into play, especially for experiments conducted in companies like Facebook or LinkedIn. \cite{gui2015} proposes an A/B testing effect analysis by taking into account network effect into a linear model. 

\par All these works contribute to the validity or efficiency of online experiments. We emphasize that our "randomness check" has a very general application, regardless of how traffic is split, the distribution of the covariates, or whether the experiments are done offline or monitored continuously. 

\section{Methodology}
\label{methodology}
\par To formulate the problem mathematically, assume the distributor
has assigned incoming users to $k$ groups. Each user is represented
by a row vector of length $m$, with each entry representing certain
profile information, such as age, area, gender, or number of times of past visit. Hence, we have $k$
data sets, each of dimension $n_i \times m, 1 \leq i \leq k$. Our goal
is to infer whether the split achieves desired randomness by testing whether the $k$ $m$-dimensional data have the same
distribution. Throughout our analysis, we assume each column is
categorical and each row (user) is independent. These are reasonable
assumptions because many important profile informations, such as
gender, age, region, are categorical. We assume the number of user in each group is large, and the network effect is negligible  compared to the scale of user numbers.

\par In the following subsections, we first state the method that is used currently by online experimenters to do the A/B split validity check \cite{gordon2016} as baseline. Then we describe a method proposed by \cite{rizzo2010}, called distance components (DISCO), of measuring
the total dispersion of the samples, which admits a partition of the
total dispersion into components analogous to the variance components in
ANOVA. We will also apply propensity method and introduce a randomized chi-square test. Each of these distribution test method has its own advantage as we show in Section \ref{simulations}. Finally, we propose a resampling technique that controls family-wise error rate (FWER) under any condition while still maintaining high power compared to some other multiplicity adjustment methods, as again shown in Section \ref{simulations}. Our contribution is to apply and modify existing testing methods to the ubiquitous A/B-multi testing validity check. 

\subsection{Baseline}
\cite{gordon2016} and some other papers apply F-test \footnote{see
  \url{https://en.wikipedia.org/wiki/F-test} for an introduction of
  F-test.} to each of the $m$ column (covariate) or some key columns of the split data sets, and claim the split is truly random if none of the p-value is smaller than 0.05. This is perhaps the most straightforward and simple way of approaching the randomization check problem, but with some potential issues. First, F-test can only test difference in mean, but cannot detect other distributional differences, such as variance (we will see in the next subsection DISCO provides a solution to this). Second, even if each column has the same distribution among the $k$ data sets, the interaction among the $m$ columns can be different for each data set, which also contribution to multi-dimensional distribution heterogeneity. Third, multiplicity adjustment is needed if the dimension $m$ is large, which is often the case for online marketing data. Otherwise, even an A/A test can have very high false positive rate.

\par To address the third issue, we apply a resampling technique (introduced later) to adjust for multiplicity. It has the advantage of controlling FWER under any dependence structure of p-values, while maintaining high power compared to other p-value adjustment methods, such as Bonferroni or Holmes method.

\subsection{DISCO one-step test}
\cite{rizzo2010} propose a new method, called distance components (DISCO), of measuring the total dispersion of the samples in multi-dimensions, which admits a partition of the
total dispersion into components analogous to the variance components in
ANOVA. They introduce a measure of
dispersion based on Euclidean distances between all pairs of sample elements,
for any power $\alpha$, which is called the index, of distances such that $\alpha \in (0,2)$. The method is based on the following key definitions and theorem.

\par Suppose that $X$ and $X'$ are independent and identically distributed (i.i.d.),
and $Y$ and $Y'$ are i.i.d., independent of $X$. If $\alpha$ is a constant such that $E||X||^\alpha < \infty$ and $E||Y||^\alpha < \infty$, define the $\mathcal{E}_\alpha$-distance (energy distance) between the distributions of $X$ and $Y$ as
\begin{equation}
\label{energy}
\mathcal{E}_\alpha(X, Y) = 2E||X - Y||^\alpha - E||X - X'||^\alpha - E||Y - Y'||^\alpha.
\end{equation}
Then we have the following theorem.
\begin{theorem}
\label{theorem1}
Suppose that $X, X' \in \mathbb{R}^p$ are i.i.d. with distribution $F$, $Y, Y' \in \mathbb{R}^p$ are i.i.d. with distribution $G$, and $Y$ is independent of $X$. If $0<\alpha \leq 2$ is a constant such tha $E||X||^\alpha < \infty$ and $E||X||^\alpha < \infty$, then the
following statements hold:\\
(i) $\mathcal{E}_\alpha(X, Y ) \geq 0.$ \\
(ii) If $0<\alpha \leq 2$, then $\mathcal{E}_\alpha(X, Y ) = 0$ if and only if $X \stackrel{\mathcal{D}}{=} Y$. \\
(iii) If $\alpha = 2$, then $\mathcal{E}_\alpha(X, Y ) = 0$ if and only if $E[X] = E[Y]$.
\end{theorem}
The proof of Theorem \ref{theorem1} can be read in \cite{rizzo2010}. Based on this theorem, DISCO statistics, which can be thought of as a variation of the ANOVA statistics, can be developed. Define the empirical distance between distributions as follows. For two $p$-dimensional samples $A = \{a_1,..., a_{n_1}
\}$ and $B = \{b_1,..., b_{n_2}
\}$, the $d_\alpha$-distance between $A$ and $B$ is defined as
\begin{equation}
\label{distance}
d_\alpha(A,B) = \frac{n_1 n_2}{n_1+n_2} [2g_\alpha(A,B)-g_\alpha(A,A)-g_\alpha(B,B)],
\end{equation}
where
\begin{equation}
\label{gfunction}
g_\alpha(A,B) = \frac{1}{n_1 n_2} \sum \limits_{i=1}^{n_1} \sum \limits_{m=1}^{n_2} ||a_i-b_m||^\alpha.
\end{equation}
Note that $d_\alpha(A,A)$ is within-sample dispersion and $d_\alpha(A,B)$ is between-sample dispersion. Similar to ANOVA analysis, we can also write total dispersion as summation of between-sample and within-sample dispersion here. Let $A_1,...,A_K$ be $p$-dimensional samples with sizes $n_1,...,n_K$. The $K$-sample $d_\alpha$-distance statistic that takes the role of ANOVA sum of squares for treatments is the weighted sum of dispersion statistics:
\begin{equation}
\label{between}
S_\alpha(A_1,...,A_K) = \sum \limits_{1 \leq j<k \leq K} \bigg(\frac{n_j+n_k}{2N}\bigg) d_\alpha(A_j,A_k).
\end{equation}
Similarly, the total dispersion of the observed response is
\begin{equation}
\label{total}
T_\alpha(A_1,...,A_K) = \frac{N}{2} g_\alpha (A,A),
\end{equation}
where $A = \sum \limits_{i=1}^K A_i$ is the pooled sample, and the within-sample dispersion is
\begin{equation}
\label{within}
W_\alpha (A_1,...,A_K) = \sum \limits_{j=1}^K \frac{n_j}{2} g_\alpha(A_j,A_j).
\end{equation}
Note that we have $T_\alpha(A_1,...,A_K) = S_\alpha(A_1,...,A_K)+W_\alpha(A_1,...,A_K)$, and when $p=1, \alpha=2$, the decomposition $T_2=S_2+W_2$ is exactly the ANOVA decomposition of the
total squared error: $SS(\text{total}) = SST + SSE$. Hence, ANOVA is a special case of DISCO method.

\par Based on the decomposition, the final statistics for testing equality of distribution is
\begin{equation}
\label{final_stat}
D_{n,\alpha} = \frac{S_\alpha/(K-1)}{W_\alpha/(N_K)},
\end{equation}
with $0<\alpha<2$ (if $\alpha=2$, the above statistics can only test equality of mean). The distribution of $D_{n,\alpha}$ is complicated so \cite{rizzo2010} uses permutation, which is a simplified version of our resampling technique, to obtain rejection thresholds. As we will see in Section \ref{simulations}, DISCO one-step test, though simple to implement, does not provide information about which dimensions are problematic when the null hypothesis of equal distribution is rejected.

\subsection{Propensity score method}
Propensity score comparison \footnote{see
  \url{https://en.wikipedia.org/wiki/Propensity_score_matching}} tests
whether the $k$ $m$-dimensional data sets have the same distribution
by fitting a model to the data to obtain the likelihood of a data point being assigned to a particular data set. To be specific, combine the $K$ data sets $A_1,...,A_K$ to be one big data set $A$, which is also the $p$-dimensional covariate space. The response is a length-$n$ vector with each entry being the class label (1,2,...,or $K$) of each data point. We can then fit logistic regression, tree-based model, support vector machine, or any class prediction model to the data to obtain a predicted label for each data point. Here we use multiple logistic regression. If the $K$ original data sets truly have the same distribution, the predicted labels should have about uniform distribution on $\{1,2,...,K\}$. We use chi-square test \footnote{see
  \url{https://en.wikipedia.org/wiki/Chi-squared_test}} on a $K \times K$ contingency table to test the uniformity of distribution. However, we can show, for logistic regression, that only when the $K$ data sets are exactly the same is the distribution of the predicted labels truly uniform. Otherwise, the predicted label will tend to be the same as the actual label due to overfitting. To resolve this, we can randomly choose a $c$ proportion of rows in each data set to be training data, and the rest are test data. But this reduces the
effective test sample size in the stage of chi-square test. For example, if we choose $p=0.8$, then only 1/5 of the data are
used to chi-square test. This can shrink power when the total number of data points are not too large. Here, like in DISCO one-step test introduced in the previous subsection, we use permutation method to get the rejection threshold for the p-value obtained from the chi-square test.

\par The propensity score method, in some sense, takes into account both marginal distributions and
interactions because all the covariates appear together on the right hand side of the link function. Interaction terms can also be added, but it is unclear how many of the interaction terms are enough since adding all the interactions are quite unfeasible when the number of covariates is large. However, as we will show in the simulation section, the propensity score method is not very sensitive to non-mean difference (such as variance difference) in distribution. Next, we will give a more direct approach to testing marginal and
interaction homogeneity.

\subsection{Randomized chi-square test}
Since the data are categorical, the distribution is determined by the
multinomial distribution parameters $p=(p_1,...,p_l)$, where each $p_i$ denotes the probability of seeing a particular combination of categories from the $m$ columns, with $\sum
\limits_{i=1}^l p_i = 1$, where $l$ is the total number of
categories. With $m$ columns in total, the total number of categories
can be huge ($2^m$ is most likely a lower bound). Theoretically, we
can create a $K \times l$ contingency table and apply chi-square test
to test whether the multinomial distribution is independent of data
sets (equivalent to the $k$ data sets having the same distribution). Yet this is not feasible under computation time
constraints. Nor can we reduce the dimensions of the table without
losing information. For example, two distributions can have the same two-way
interaction but different three-way interaction. However, if we have
further information about the number of categories of each column, we
can reduce $l$ accordingly. A simple example is that all columns are
binary (two categories). The distribution in this case is determined
by marginal distribution and two-way interaction. Without further
information,  we compromise by
introducing randomization: 
choosing a small subset of columns each time, say $C$ columns, to do
chi-square test, and then repeat the process $D$ times. The argument
is that even if the
non-null columns are sparse, by repeatedly choosing $D$ times for
relatively large $D$, the non-null columns still have large
probability of being picked at least once. Assume there are $m$ columns in total, the probability that one single column is not picked even once is $(1-C/m)^D$, which is approximately $1-CD/m$ when $C/m$ is small. Thus, $C$ and $D$ have equal impacts on the selection probability. In practice, increasing $D$ costs less computation than increasing $C$.

\par Note that randomized chi-square only applies to categorical data among the three methods we introduce. In practice, we can ``categorize'' continuous data into buckets, which sometimes gives us higher heterogeneous distribution detection power than applying the other methods directly on the continuous data.

\subsection{A resampling technique}
Resampling \footnote{see
  \url{https://en.wikipedia.org/wiki/Resampling_(statistics)}}  is
often used when the distribution of the test statistics is not known
or hard to derive. In our randomized chi-square test, for example, the
distribution of P-values depends on the unknown distribution of the
original data sets. Thus, we can use resampling to set the threshold
of rejection. If we are in single hypothesis testing scenario (one-step DISCO and propensity score methods), our proposed resampling technique is equivalent to permutation test (see Algorithm 1) \footnote{In all the algorithms here, without loss of generality we assume the computed statistics are P-values.}.

\begin{algorithm} 
  \begin{algorithmic}[1]
    \State Get original statistics $P_0$ (e.g. from propensity score method).
    \State Randomly permute rows (users) among the $k$ data sets, and calculate the P-value from the permuted data.
    \State Repeat step 2 $B$ times, and we obtain $B$ P-values from resampling, denoted by a vector $\tilde{P}^* = (\tilde{P}^*_1,...,\tilde{P}^*_B)$.
    \State Choose the threshold $t$ as the $\alpha$(5) percentile of $\tilde{P}^*$, and reject the null hypothesis if $P_0<t$.
  \end{algorithmic} 
  \caption{resampling for single hypothesis}
  \label{alg:algorithm1}
\end{algorithm}
It is easy to see that Algorithm 1 controls FWER exactly, because all permutations have equal probability under the null. Under multiple hypotheses testing scenario, we modify Algorithm 1 a little (see Algorithm 2).

\begin{algorithm} 
  \begin{algorithmic}[1]
    \State Assume there are $m$ hypotheses. Get the original vector of statistics, denoted $\hat{P}_0=(P_{01},...,P_{0m})$.
    \State Randomly permute rows (users) among the $k$ data sets, and calculate the vector of P-values from the permuted data, denoted by $\tilde{P}^*_1 = (\tilde{P}^*_{11},...,\tilde{P}^*_{1m})$. Let $Pmin_1 = \min(\tilde{P}^*_1)$.
    \State Repeat step 2 $B$ times, and we obtain $B$ minimum P-values from permuted data, denoted by a vector $\tilde{P}min = (Pmin_1,...,Pmin_B)$. 
    \State Choose the threshold $t$ as the $\alpha$(5) percentile of $\tilde{P}min$, and reject the null hypothesis i if $P_{0i}<t$.
  \end{algorithmic} 
  \caption{resampling for multiple hypotheses}
  \label{alg:algorithm2}
\end{algorithm}
We claim that Algorithm 2 controls FWER under any hypotheses configuration.
\begin{theorem}
Assume P-values are marginally Uniform(0,1). Under any covariance structure, Algorithm 2 controls FWER, so it controls FDR as well.
\end{theorem}
\begin{proof}
The proof is conceptually simple. Assume there are $r$ non-null and $m-r$ null hypotheses among the $m$ hypotheses. Then
the probability of making at least one false rejection is the probability that the minimum of a subset of $m-r$ null P-values is less than the $\alpha$th percentile of the distribution of the minimum of total $m$ null P-values, which is less than $\alpha$.
\end{proof}

\par In practice, resampling can also be applied after hypotheses selection (see Algorithm 3).
\begin{algorithm} 
  \begin{algorithmic}[1]
    \State Assume there are $m$ hypotheses. A selection rule $S$ selects $s$ hypotheses with statistics $\hat{P}^S=(P_{m1},...,P_{ms})$.
    \State Randomly permute rows (users) among the $k$ data sets, and apply selection rule S, and denote selected P-values by $\tilde{P}^*_1 = (\tilde{P}^*_{11},...,\tilde{P}^*_{1s1})$. Let $Pmin_1 = \min(\tilde{P}^*_1)$.
    \State Repeat step 2 $B$ times, and we obtain $B$ minimum selected P-values from permuted data, denoted by a vector $\tilde{P}min = (Pmin_1,...,Pmin_B)$. 
    \State Choose the threshold $t$ as the $\alpha$(5) percentile of $\tilde{P}min$, and reject the null hypothesis i in the originally selected $s$ P-values if $P_{mi}<t$. 
  \end{algorithmic} 
  \caption{resampling after hypotheses selection}
  \label{alg:algorithm3}
\end{algorithm}
Empirical results show that Algorithm 3 also controls FWER, thus FDR, when the selection rule $S$ satisfies certain properties. That is, when non-null hypotheses are more likely to be selected than null hypotheses. Intuitively, this controls FWER because the minimum of a subset of a total $s$ null P-values is stochastically larger than the minimum of all the $s$ null P-values. 

\par We will compare this resampling technique to other multiplicity adjustment methods such as Holm's procedure \footnote{See \url{https://en.wikipedia.org/wiki/Holm0-Bonferroni_method} for Holm's procedure.}, and the Benjamini-Yekutieli (BY) procedure proposed in \cite{benjamini2001} in the next section. 
 
\section{Simulations and analysis}
\label{simulations}
Since real data sets are huge and messy, we first use simulations to compare
our methods introduced in the previous section. Heterogeneity in multi-dimensional distributions has two compositions: (a) heterogeneity in marginal distribution of one (or more) dimension(s); (b) heterogeneity in the covariance (interaction) structure among dimensions. In the following simulations, the data is 
generated with either (a) or (b), or both. 

\par For randomized chi-square method, deciding the choices of the maximum number of columns to sample each time ($C$) and the number of times to sample ($D$) is tricky. When dimensions become large, we hope to increase $C$ to capture more complicated interaction structure. Yet the number in each cell of the $R \times C$ table will decrease, which will diminish the power of chi-square test. Empirically, when the number of columns is not too large, and if the sum of each column of the $R \times C$ table is greater than or equal to 5, we pick $C$ between $1/10 m$ and $1/5 m$, where $m$ is the dimension of the data sets. 

\par We do not do any feature selection to reduce dimension in our simulations, because in reality online marketing data has a great range of variety depending on the company, and each company has its own way of trimming data sets (either with general feature selection technique such as LASSO, or with domain knowledge). Thus, we do not apply Algorithm 3 in the previous section in this paper.

\subsection{Detection of heterogeneity in marginal distribution or interaction structure}
In this first simulation, we consider detection of marginal distribution difference and interaction structure difference separately. We simulate four data sets, each with 10 columns (dimensions) and 100 rows. Three of them have iid entry with multinomial distribution with probability vector (0.25,0.25,0.25,0.25) (so the number of categories of each column is also four). 

\par In Scenario 1, the last data set has different marginal distribution than the other three data sets: for weak signal, the probability vector is (0.3,0.25,0.25,0.2); for medium signal, it is (0.4,0.25,0.2,0.15); for strong signal, it is (0.5,0.2,0.2,0.1). For each heterogeneity we also vary the number of columns (from 1 to 10) that have heterogeneous marginal distribution to test the sensitivity of each method. 

\par In Scenario 2, the last data set has different interaction structure among columns keeping the marginal distribution the same as the other three sets (the other three have independent columns): first sort each column, and then for strong signal, each column is rotated by 10 (turning an array $A$ to $A[10:] + A[:10]$. Here $+$ denotes concatenation) from the previous column, so the columns are positively correlated; for medium signal, after rotation 40\% of the data in each column is permuted to mitigate the correlation; for weak signal, after rotation 80\% of each column is permuted randomly. Like in Scenario 1, we also vary the number of columns (from 2 to 10) with heterogeneous interaction structure. We compare
four methods: baseline t-test, propensity score method with resampling, DISCO method, and randomized chi-square test with resampling. 

\begin{figure}[H]
\includegraphics[scale=0.5]{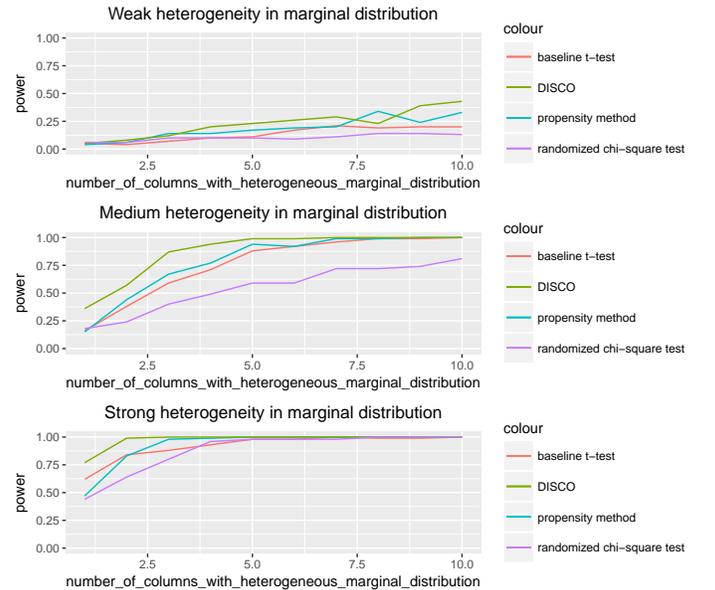}
\caption{All three plots are the results of detection of heterogeneity in marginal distribution, the top one for small difference, the middle one for medium difference, the bottom one for big difference}
\label{plot1}
\end{figure}
\begin{figure}[H]
\includegraphics[scale=0.5]{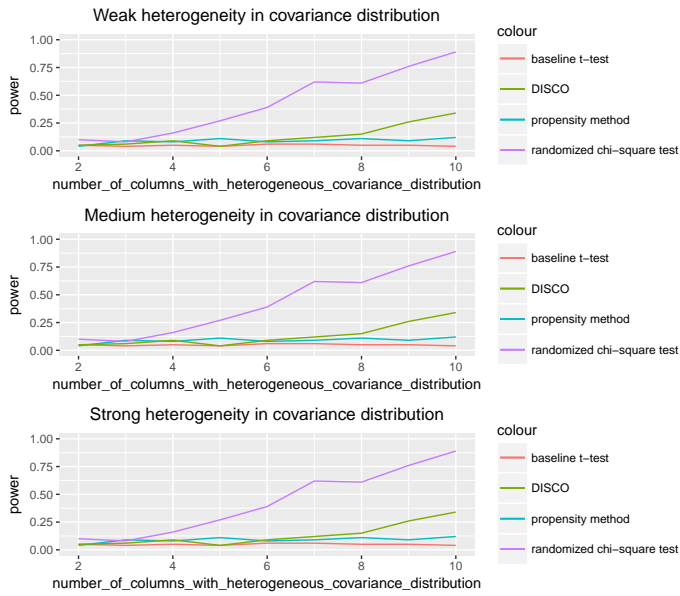}
\caption{All three plots are the results of detection of heterogeneity in interaction structure, the top one for weak correlation, the middle one for medium correlation, the bottom one for strong correlation}
\label{plot2}
\end{figure}

\par From Figure \ref{plot1} and \ref{plot2} we see that baseline t-test has high power for detecting marginal heterogeneity, but not for detecting interaction heterogeneity. This agrees with our conjecture. Thus, in the following simulations, we will not test this method anymore. DISCO and propensity score methods have relatively higher power in detecting marginal difference, while randomized chi-square method has significant advantage in detecting interaction heterogeneity. Furthermore, while DISCO or propensity score methods can only tell whether the multi-dimensional distributions are the same, randomized chi-square test can also flag individual problematic columns when the distributions are imbalanced. The flagged columns may not be exhaustive due to randomness in column selection, but it definitely provides a starting point for the follow-up detail diagnosis procedure.  

\par We also see from the power plots that interaction heterogeneity is in general harder to detect than marginal distribution heterogeneity since the former has convex power line (power only increase significantly when the number of heterogeneous columns is large) while the latter has concave power line.

\subsection{Varying dimension}
We next vary the dimension of data sets from 10 to 50 while holding the number of rows to be 100. We compare the three methods, propensity methods, DISCO, and randomized chi-square test, on weak heterogeneity (weak marginal heterogeneity + weak interaction heterogeneity as defined in the previous simulation), medium heterogeneity (weak marginal heterogeneity + medium interaction heterogeneity), and strong heterogeneity (medium marginal heterogeneity + medium interaction heterogeneity). In each of the three signal level the number of heterogeneous columns is 1/5 of the dimensions.  

\begin{figure}[H]
\includegraphics[scale=0.5]{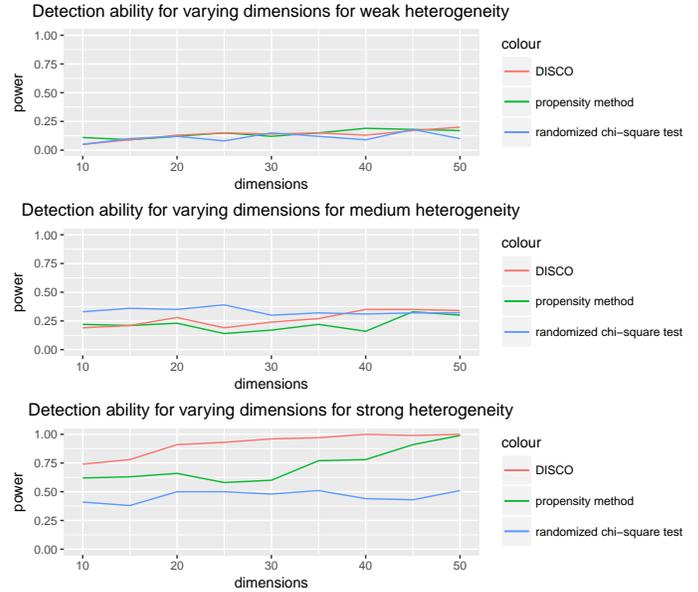}
\caption{All three plots are the results of detection power of heterogeneity as dimension increases, the top one for weak heterogeneity, the middle one for medium heterogeneity, the bottom one for strong heterogeneity}
\label{plot3}
\end{figure}

\par From Figure \ref{plot3} we see randomized chi-square behaves almost the same, or a little better, when the difference in distribution is not too large, while the other two methods, especially DISCO, behave significantly better when the dimension of data sets is relatively large. In other words, randomized chi-square test has higher power in detecting many small effects.

\par We also notice that compared to the other two methods, increasing the dimension of the data has little positive effect in power for randomized chi-square test. The reason is that since the portion of heterogeneous columns is unchanged as dimension increases, the probability of picking the "problematic columns" is also relatively unchanged as dimension increases. For the other two methods, the power has an increasing trend as dimension increases.

\par We use resamling procedure to get threshold of rejection in all the four methods, so here we also compare the power of resampling procedure, and one popular multiplicity adjustment procedure, the Holm's procedure \footnote{See \url{https://en.wikipedia.org/wiki/Holm\%E2\%80\%93Bonferroni_method} for an introduction of Holm's procedure.} Both resampling and Holm's procedure can control type I error under any P-value structure. 

\begin{figure}[H]
\includegraphics[scale=0.4]{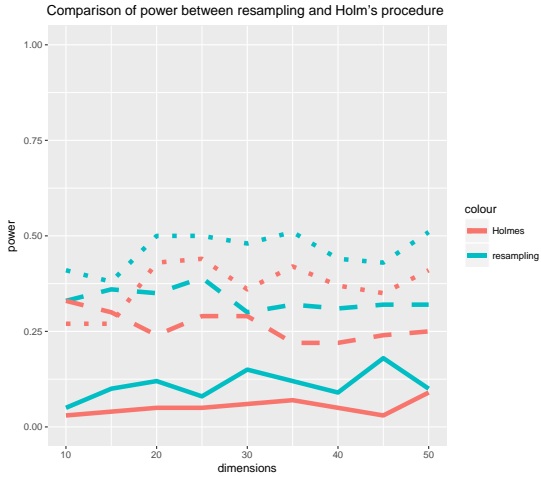}
\caption{Power comparison of resampling and Holm's procedure for randomized chi-square method, the top one for weak heterogeneity, the middle one for medium heterogeneity, the bottom one for strong heterogeneity}
\label{plot4}
\end{figure}

Figure \ref{plot4} plots the power comparison of applying randomized chi-square test + resampling and applying randomized chi-square test + Holm's procedure. We see that for any heterogeneity level, resampling has higher power than Holm's procedure, which justifies the use of resampling procedure in multiple hypotheses testing problems. 

\subsection{Simulated real-world scenario}
Next, we simulate a California-local real online marketing A/B testing scenario. Assume the company is an online retailer for fashion, skin care, and cosmetics. The data is collected from a 24-hour window from 12AM February 21st to 12AM February 22nd. There are 8 columns representing area of living, browser type, gender, age, employment status, income, accumulated number of visits (in the past year), and converted or not before respectively. Table 1 below summarizes the encoding details of these 8 features. 

\begin{table}[H]
\begin{center}
\begin{tabular}{ |c|c| } 
 \hline
 Column names & Encoding  \\ \hline
 Area of living & 0: south Cal; 1: north Cal; 2: mid Cal \\ \hline
 Browser type & 0: Chrome; 1: Safari; 2: Firefox; \\ 
 & 3: Internet Explorer; 4: others \\ \hline 
Gender &  0: male; 1: female \\ \hline 
Age & 0: <20; 1: 20-30; 2: 30-40;  \\ 
& 3: 40-50; 4: >50 \\ \hline
Employment status & 0: student; 1: employed; 2: unemployed \\ \hline
Income & 0: <50,000; 1: 50,000-100,000;  \\ & 2: 100,000-200,000; 3: >200,000 \\ \hline
Accumulated & 0: <3; 1: 3-10; 2: 10-20; 3: >20\\ 
number of visits &  \\ \hline
Converted before & 0: no; 1: yes\\ \hline
\end{tabular}
\label{table1}
\caption{Column names and encodings} 
\end{center}
\end{table}
The number of categories for each column ranges from 2 to 5. The real interaction structure among these 8 columns can be complicated. We simplify this relationship and summarize it in Figure 5 below. Area of living and Browser type are independent variables, while the other 6 variables have complicated correlations. It is by no means accurate. For example conditioning on employment status, age and gender may not be independent in reality. We consider the scenario where the data traffic is not randomly split for a period of time. For example, almost all traffic is assigned to set A for two hours in the morning, which results in set A has more data than set B. To balance the number of data points in the two sets, the experimenter assigns most of the traffic in the evening to set B. A direct result is that employment status has different distributions in set A and B --- unemployed people take up much higher proportion in the morning time when employed people and students are busy working than they do in the evening time.

\begin{figure}[H]
\centering
\begin{tikzpicture}[auto,node distance=1.5cm]
  %
  \node[entity] (node1) {Employment status}
  [grow=up,sibling distance=2cm];
    \node[relationship] (rel3) [above = of node1] {Income};
     \node[relationship] (rel2) [below left = of node1] {Gender};
  \node[relationship] (rel1) [below right = of node1] {Age};
  \node[entity] (node2) [above right = of rel1, xshift=-1.6cm]	{Accumulated number of visits};
  \node[relationship] (rel4) [above = of node2] {Converted or not};
  \node[entity] (node3) [above  = of rel3]{Area of living};
  \node[entity] (node4) [right = of node3]{Browser type};
  \path (rel1) edge node {1} (node1)
  edge	 node {4}	(node2)
  (rel2) edge node {2} (node1)
  edge	 node {4}	(node2)
  (rel3) edge node {3} (node1)
  edge	 node {4}	(node2)
  (rel4) edge node {5} (node2)
    (rel4) edge node {5} (rel3);
\end{tikzpicture}
\caption{Interaction structure among the 8 columns}
\label{diagram}
\end{figure}
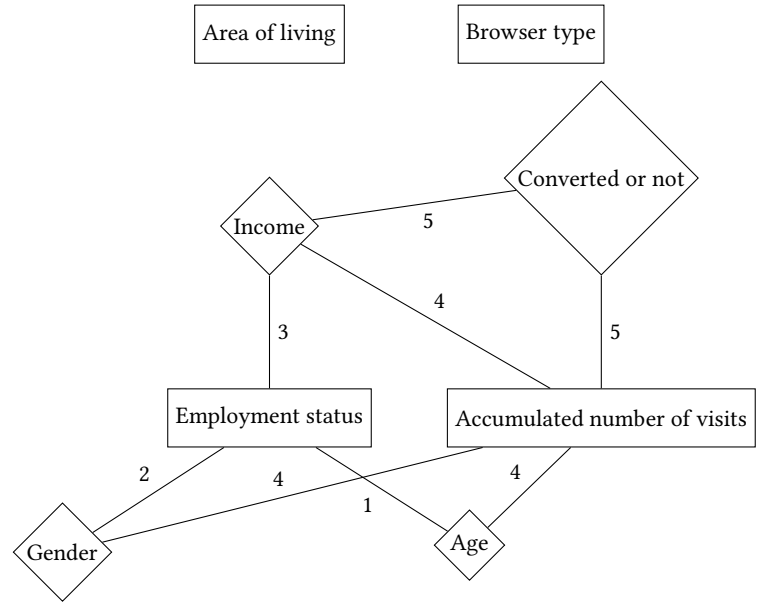

Figure \ref{diagram} also shows a way to generate our data. All marginal distributions are multinomial. Area of living and Browser type are independent which can be easily generated. Then we generate the Employment status with certain probability vector. Conditioning on Employment status, we generate age (the edge with 1), gender (the edge with 2), and income (the edge with 3). Then conditioning on income, gender, age, we generate Accumulated number of visits (the edges with 4). Finally we generate Converted or not conditioning on income and Accumulated number of visits (edges with 5). A detailed data generation procedure is summarized in the Appendix. In particular, the probability vector for Employment status is (0.3,0.3,0.4) in set A (more people in the morning) and (0.4,0.4,0.3) in set B (more people in the evening), which results in difference in marginal distribution and correlation structure in other dimensions. We control the number of data points to be 800 in either data sets.

\par The data is simulated 100 times. The power for the propensity method, DISCO method, and randomized chi-square test is 0.93, 0.91, and 0.87 respectively. All three methods can detect the multi-dimensional distribution heterogeneity with high power. Again there is a trade-off: randomized chi-square method has comparably lower power but can point out the most imbalanced columns. 

\begin{table}[H]
\begin{center}
\begin{tabular}{ |c|c|| } 
 \hline
 Column names  & Number of times being rejected \\ \hline 
 Area of living & 28 \\ \hline 
 Browser type & 23 \\ \hline 
 Employment status &163 \\ \hline
 gender & 43 \\ \hline
 age & 85 \\ \hline
 income & 44 \\ \hline
 Accumulated number of visits & 50 \\ \hline
 Converted or not & 43\\ \hline
\end{tabular}
\caption{the number of time each column gets rejected in 100 simulations with 10 column sampling times per simulation}  
\label{table2}
\end{center}
\end{table}
Table \ref{table2} displays the number of time each column gets rejected in 200 simulations with $S=10$ (sampling columns 10 times per simulation) and $C=3$ (sampling maximum 3 columns each time). We see Area of living and Browser type have significantly fewer times of being rejected because they are balanced and independent. Employment status has the largest number of sampling times since it is directly affected by the timing of the day, with the rest 5 columns having milder imbalance depending on its correlation with Employment status.

\section{Anonymous real data}
\label{realdata}
We also try our tests on some auto-personalization datasets provided by Adobe Digital Marketing Cloud. Auto-personalization is personalized recommendation to individual customers learned from past customer behavior in the database. Usually a random A/B splitting is needed to serve as baseline to evaluate any recommendation methodology. We obtain such a pair of randomly-split data sets, and test the propensity method, DISCO method, and randomized chi-square test on a data set from the online marketing section of a company. The data is split into set A and B, each of dimension $5055 \times 383$, with 2 to 50 categories per column. For privacy concern, we do not get to know the encodings of the columns and the detailed data collecting procedure. However, propensity score with resampling, DISCO method, and
randomized chi-square tests all reject the hypothesis that the two
distributions are the same. Randomized chi-square also provides some combinations of columns that are the most imbalanced, one of which is columns 196, 57, 248, 260, 271, 342, 374, and 239. These 8 columns have 50 different combinations (for example if one column can take on 2 values, another 3 values, then the two columns together can take on at most $2 \times 3 = 6$ values). Figure \ref{plot5} shows a bar plot of the 15 combinations that have the most counts from both data sets, and a ratio of counts in set B to those in set A. We see the two categorical distributions do differ a great deal, with counts ratio being as high as 3. If we get more information on the these two sets we can do further analysis on these selected columns to identify the source of heterogeneity. 

\begin{figure}[H]
\includegraphics[scale=0.45]{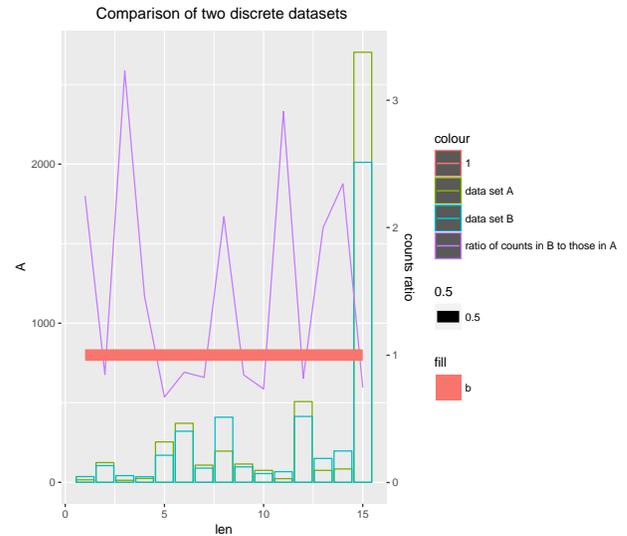}
\caption{comparison of a subset of columns from two categorical data sets, the red bar representing counts from set A, the green bar from set B, purple line denoting ratio of counts in A to counts in B}
\label{plot5}
\end{figure}

\section{Discussion}
\label{discussion}
In summary, A/B testing has a wide range of application in online marketing and other areas, so making sure the test indeed provides valid results is crucial. This paper formulates the validity of A/B testing results as a hypothesis testing problem in multi-dimensional space. We propose three ways to test whether two or more data sets come from the same distribution and each has its own advantage. The propensity score and DISCO methods can both be generalized to continuous data, but randomized chi-square test is only applicable to categorical data. Of course, the simplest way is to categorize the continuous data. In fact, even when the data is one-dimensional normal with the same mean but difference variance, both propensity and DISCO methods have very low power, but categorize each data point $x$ to be the floor of $2x$ and apply randomized chi-square test yields high power.

\par Besides testing equivalence of distribution after the experiments have been carried out, we also hope to make sure the online experiments are correctly carried out from the beginning. This is a meaningful and promising area of research that involves multivariate analysis, multiple testing, feature selection, observational data and sequential analysis. There has been active research going on in some parts of this big area, including the papers we mentioned in Section 2. We also need to work closely with data engineers to effectively implement these methodology.


\appendix
\section*{Appendix}
\subsection{Real-world simulation data generation process}
\begin{itemize}
\item 
Generate Area of living $\sim$ Multinomial ((0.33,0.33,0.34)).
\item
Generate Browser type $\sim$ Multinomial ((0.3,0.3,0.2,0.15,0.05)).
\item
For data set A, generate Employment status $\sim$ Multinomial ((0.3,0.3,0.4)); for data set B, generate Employment status $\sim$ Multinomial ((0.4,0.3,0.3)).
\item
gender $\bigg|$ Employment status $= 0$ $\sim$ Multinomial  ((0.5,0.5)); gender $\bigg|$ Employment status $= 1$ $\sim$ Multinomial  ((0.6,0.4)); gender $\bigg|$ Employment status $= 2$ $\sim$ Multinomial  ((0.3,0.7)).
\item
age $\bigg|$ Employment status $= 0$ $\sim$ Multinomial  ((0.7,0.2,0.05,0.03,0.02)); age $\bigg|$ Employment status $= 1$ $\sim$ Multinomial  ((0.1,0.2,0.3,0.3,0.1)); age $\bigg|$ Employment status $= 2$ $\sim$ Multinomial  ((0.1,0.02,0.04,0.04,0.8)).
\item
income $\bigg|$ Employment status $= 0$ $\sim$ Multinomial  ((0.7,0.2,0.1,0)); income $\bigg|$ Employment status $= 1$ $\sim$ Multinomial  ((0.2,0.3,0.3,0.2)); income $\bigg|$ Employment status $= 2$ $\sim$ Multinomial  ((0.8,0.15,0.05,0)).
\item
ANOV (Accumulated number of visits) $\bigg|$ (income $\leq 2$ \& age $\geq 3$) $\sim$ Multinomial  ((0.7,0.15,0.1,0.05)); ANOV $\bigg|$ (income $\leq 1$ \& gender $=0$ \& age $\leq 2$) $\sim$ Multinomial  ((0.5,0.2,0.2,0.1)); ANOV $\bigg|$ (income $\geq 3$ \& gender $=1$ \& age $> 2$) $\sim$ Multinomial  ((0.15,0.2,0.25,0.4)); ANOV $\bigg|$ other combinations of age, gender and income $\sim$ Multinomial ((0.25,0.25,0.25,0.25)).
\item
Conv (converted or not) $\bigg|$ (income $< 2$ \& ANOV $< 2$) $\sim$ Multinomial  ((0.9,0.1)); Conv $\bigg|$ (income $> 2$ \& ANOV $< 2$) $\sim$ Multinomial  ((0.85,0.15)); Conv $\bigg|$ (income $> 2$ \& ANOV $> 2$) $\sim$ Multinomial  ((0.65,0.35)); Conv $\bigg|$ (income $< 2$ \& ANOV $> 2$) $\sim$ Multinomial ((0.8,0.2)).
\end{itemize}
